\documentclass[11pt]{article}
\ifdefined\pdfinfoomitdate\pdfinfoomitdate=1\fi
\ifdefined\pdftrailerid\pdftrailerid{}\fi
\ifdefined\pdfsuppressptexinfo\pdfsuppressptexinfo=15\fi
\usepackage[T1]{fontenc}
\usepackage{lmodern}
\usepackage[margin=1in]{geometry}
\usepackage{amsmath,amssymb,amsthm}
\usepackage{booktabs}
\usepackage[hidelinks]{hyperref}
\hypersetup{
  pdftitle={Exact Sizewise Least Cores for Candidate-Disjoint Approval Chains},
  pdfauthor={Jiarui Fang},
  pdfsubject={Approval-based committee elections, core stability, and the sizewise least core},
  pdfkeywords={approval-based committee elections, core stability, sizewise least core, candidate-interval preferences, dynamic programming}
}
\newtheorem{theorem}{Theorem}
\newtheorem{lemma}{Lemma}
\newtheorem{proposition}{Proposition}
\newtheorem{corollary}{Corollary}
\theoremstyle{definition}\newtheorem{example}{Example}
\theoremstyle{remark}
\newcommand{\eps}{\varepsilon}

\title{Exact Sizewise Least Cores for\\
Candidate-Disjoint Approval Chains}
\author{
Jiarui Fang\\
Boston University\\
Boston, Massachusetts, USA\\
\texttt{baymin@bu.edu}\\
\href{https://orcid.org/0009-0006-9100-0445}
     {ORCID: 0009-0006-9100-0445}
}
\date{3 September 2026}
\begin{document}
\maketitle
\begin{abstract}
For approval-based committee elections, the sizewise least-core value
$\Gamma(P)$ minimizes the maximum excess support of a deviation over
size-$k$ committees.  Deviations are nonempty and contain at most $k$
candidates, with support measured relative to their Hare quota.
A negative value certifies core nonemptiness, while $-\Gamma(P)$ is the
largest worst-case quota shortfall secured by any committee.
We quantify this margin on candidate-disjoint approval chains, a domain
whose qualitative core nonemptiness is already known.
We first factor each committee's deviation score over disjoint components.
For chains, a closed local table yields the exact value and an attaining
committee in polynomial time for rational weights.
For nonnegative normalized real weights, let $q_{\mathrm{act}}\ge1$ count
the connected components of the positive-support incidence graph.  Then
\[
 \Gamma(P)\le-\frac{1}{k\min\{q_{\mathrm{act}},k+1\}}.
\]
Explicit extremal profiles show that this dependence on the active-component
count is sharp.  The margin is $1$-Lipschitz in total variation on the full
ballot space.  Each negative margin therefore certifies a neighborhood with
nonempty core, even when perturbations leave the chain domain.
\end{abstract}

\section{Introduction}
The core rules out a coalition that can afford a strictly preferred
deviation at the Hare quota.  Its existence for arbitrary approval elections
remains open, despite progress for restricted parameters and preference
domains.  We isolate a decomposable chain domain while retaining the signed
excess support of the best deviation.  The central argument is that exact
prefix-packed local tables make the component reduction polynomial-time and
expose the sharp dependence on the intrinsic number of support-active
components.

For intuition, consider disjoint local candidate lists whose approval sets are
nested within each list.  Voters may then differ only in how far down one list
they approve.  A qualitative theorem says that some committee is unblocked.
It does not determine which committee maximizes the minimum shortfall of
deviation support below its Hare quota.
For a core-nonempty profile, $-\Gamma(P)$ is this largest guaranteed gap,
measured as a fraction of the electorate.
Our tables allocate seats across lists and compute that gap exactly.
The total-variation bound then turns it into a sufficient stability
certificate for arbitrary changes in the ballot distribution.
This certificate is not asserted to equal the optimal distance to core
emptiness.  The intrinsic count $q_{\mathrm{act}}$
ignores empty ballots and unused candidates, and controls a guaranteed
stability buffer.

The contribution has four layers.  First, a coalition-cost identity and a
careful over-budget argument give exact component factorization.  Second,
coordinatewise monotonicity reduces the global optimum to local tables and a
seat-allocation dynamic program.  Third, nested approvals make every local
table explicit, eliminating exponential local enumeration.  Fourth, the
closed table gives an order-statistic formula and heap reconstruction, and
converts demand minimality into the optimal
$\min\{q_{\mathrm{act}},k+1\}$ bound and matching extremal profiles.  A
fixed-candidate refinement separates the forced case $m=k$, while a
total-variation inequality extends each negative margin to nearby arbitrary
approval distributions.

Several boundaries are essential.  A deviation without strict supporters
attains the baseline $-1/k$.  Selected coalition mass need not equal the
actual supporter mass of its minimum-cost deviation.  Empty ballots,
zero-weight types, and null candidates must not enter $q_{\mathrm{act}}$.
Rational arithmetic yields a polynomial-time algorithm,
but no strong-polynomial bit bound is claimed.  Computation below is
falsification and regression evidence only; all quantitative results are
proved analytically.

\section{Model and coalition costs}
Let $C$ be finite, $1\le k\le|C|$, and let $N$ be a finite set of voter
types.  Type $i$ has approval set $A_i\subseteq C$ and weight $p_i\ge0$,
where $\sum_i p_i=1$.  Write $u_i(X)=|A_i\cap X|$.  A committee is a set
$W\subseteq C$ of size $k$, and an admissible deviation is a nonempty
$T\subseteq C$ with $|T|\le k$; it may overlap $W$.  Define
\[
 R_W(T)=\{i:u_i(T)>u_i(W)\},\qquad
 g_p(W,T)=p(R_W(T))-|T|/k,
\]
\[
 g_p(W)=\max_Tg_p(W,T),\qquad
 \Gamma(P)=\min_{|W|=k}g_p(W).
\]
For rational weights and $k<|C|$, $\Gamma(P)$ coincides with Berker et al.'s
fixed-profile sizewise-least-core value $\mu_{\mathbf{x}}^*$
~\cite[Section~3]{berker2026}.
Here $\mathbf{x}(A)=\sum_{i:A_i=A}p_i$ aggregates identical approvals;
the same definition also covers real weights and $k=|C|$.
In particular,
\[
 \begin{aligned}
 W\text{ is core-stable}&\iff g_p(W)<0,\\
 P\text{ has nonempty core}&\iff\Gamma(P)<0.
 \end{aligned}
\]
Blocking includes equality at the Hare quota, while utility improvement is
strict.
For a chosen committee, a value $g_p(W)=-\eta<0$ means that every deviation's
support falls short of its sizewise quota by at least $\eta$; larger $\eta$
is a stronger stability buffer.  For example, a singleton contained in $W$
has no strict supporter and score $-1/k$, whereas a deviation with support
exactly $|T|/k$ blocks at score zero.

For nonempty $S\subseteq N$, let $\tau_W(S)$ be the minimum size of an
admissible deviation that strictly improves every type in $S$, and set it to
$+\infty$ when none exists.

\begin{theorem}[Exact coalition-cost identity]\label{thm:coalition}
For every size-$k$ committee $W$,
\[
 g_p(W)=\max\left\{-\frac1k,
 \max_{\substack{\varnothing\ne S\subseteq N\\\tau_W(S)<\infty}}
 \left(p(S)-\frac{\tau_W(S)}k\right)\right\}.
\tag{1}
\]
\end{theorem}
\begin{proof}
If $R_W(T)=\varnothing$, the score is $-|T|/k\le-1/k$.  Otherwise
$\tau_W(R_W(T))\le|T|$, which gives one inequality.  Conversely, a
minimum-cost deviation for $S$ may improve a strict superset of $S$, but its
actual supporter mass is at least $p(S)$.  Finally, for any $w\in W$, the
singleton $\{w\}$ has no strict supporter and attains $-1/k$.
\end{proof}

\section{Exact component factorization}
Assume
\[
 C=C_1\sqcup\cdots\sqcup C_h,\qquad
 N=N_1\sqcup\cdots\sqcup N_h,
\]
where every $C_j$ is nonempty and $A_i\subseteq C_j$ for $i\in N_j$.
We call the pairs $(C_j,N_j)$ the declared blocks.  They are input metadata,
not the intrinsic support-active components defined later, because empty
ballots and null candidates remain assigned to declared blocks.
Every candidate approved by no type is assigned to an existing block before
the partition is fixed; it creates no voter-bearing component but remains
counted in the local seat allocation.  Write $W_j=W\cap C_j$.  For nonempty
$S\subseteq N_j$, let $\tau_j(W_j,S)$ be the minimum size of a nonempty
$T_j\subseteq C_j$, $|T_j|\le k$, that strictly improves every type in $S$.
Put
\[
 \lambda_j(W_j)=
 \max_{\substack{\varnothing\ne S\subseteq N_j\\
                  \tau_j(W_j,S)<\infty}}
 \left(p(S)-\frac{\tau_j(W_j,S)}k\right),
\tag{2}
\]
with value $-\infty$ if no nonempty local coalition can improve.  For
extended-real $z=(z_1,\ldots,z_h)$ define
\[
 H(z)=\max\left\{-\frac1k,
      \max_{\varnothing\ne J\subseteq[h]}\sum_{j\in J}z_j\right\},
\tag{3}
\]
where a sum containing $-\infty$ is $-\infty$.

\begin{theorem}[Exact component factorization]\label{thm:factor}
For every size-$k$ committee $W$,
\[
 g_p(W)=H\bigl(\lambda_1(W_1),\ldots,\lambda_h(W_h)\bigr).
\tag{4}
\]
\end{theorem}
\begin{proof}
Split a deviation into $T_j=T\cap C_j$.  Component supporter sets are
disjoint and costs add.  Delete pieces without local supporters.  Every
remaining local score is at most $\lambda_j(W_j)$, while a deviation with no
supporter scores at most $-1/k$.  This proves the upper bound.

For the reverse bound, fix a nonempty $J$ with finite coordinates.  Choose
for each $j\in J$ an attaining coalition $S_j$ in (2) and a minimum-cost
local deviation of size $c_j=\tau_j(W_j,S_j)$.  Then
\[
 \lambda_j(W_j)=p(S_j)-c_j/k.
\]
Here $p(S_j)$ is the selected coalition mass, not necessarily the actual
supporter mass of its deviation.  If $\sum_jc_j\le k$, the united deviation
is feasible and its actual supporter mass is at least
$\sum_jp(S_j)$, so its score is at least the selected subset sum.  If
$\sum_jc_j>k$, the union is not asserted to be feasible.  Instead,
$\sum_jc_j\ge k+1$ and the disjoint type blocks give
$\sum_jp(S_j)\le1$, whence
\[
 \sum_{j\in J}\lambda_j(W_j)
 \le1-(k+1)/k=-1/k.
\]
The baseline in (1) is already attained.  This proves (4), including
$\lambda_j=-\infty$, local saturation, zero weights, $h=1$, $k=1$, and
$k=|C|$.
\end{proof}

\section{Exact component dynamic programming}
Put $b_j=|C_j|$, $B_j=\min\{k,b_j\}$, and
\[
 F_j(s)=\min_{\substack{W_j\subseteq C_j\\|W_j|=s}}\lambda_j(W_j),
 \qquad0\le s\le B_j.
\tag{5}
\]
Store a local argmin with each entry.

\begin{theorem}[Exact component DP]\label{thm:dp}
\[
 \Gamma(P)=
 \min_{\substack{s_1+\cdots+s_h=k\\0\le s_j\le B_j}}
 H\bigl(F_1(s_1),\ldots,F_h(s_h)\bigr).
\tag{6}
\]
The union of local argmins for an optimal allocation is an attaining global
committee.
\end{theorem}
\begin{proof}
$H$ is coordinatewise nondecreasing.  For a fixed seat allocation, every
compatible committee has local coordinate at least $F_j(s_j)$, and
independent local argmins attain all coordinates simultaneously.  Theorem
\ref{thm:factor} and minimization over allocations prove the result.
\end{proof}

The useful exact sign identity is
\[
 H(z)=\begin{cases}
 \sum_j\max\{z_j,0\},&\text{if some }z_j>0,\\
 \max\{-1/k,\max_jz_j\},&\text{if every }z_j\le0.
 \end{cases}
\tag{7}
\]
The baseline in the second branch cannot be omitted.  Given arbitrary local
tables, first restrict to all-nonpositive allocations.  Let $\beta_j(u)$ be the
smallest bottleneck after assigning $u$ seats among the first $j$ blocks.
Initialize $\beta_0(0)=-\infty$ and all other states to $+\infty$, and use
\[
 \beta_j(u)=\min_{\substack{0\le s\le\min\{B_j,u\}\\F_j(s)\le0}}
 \max\{\beta_{j-1}(u-s),F_j(s)\},
 \qquad 0\le u\le k.
\]
An empty minimum is $+\infty$.  If $\beta_h(k)<+\infty$, the answer is
$\max\{-1/k,\beta_h(k)\}$, which no positive-score allocation can improve.
Otherwise every feasible allocation has a positive entry.  Initialize
$\alpha_0(0)=0$ and all other states to $+\infty$, and compute
\[
 \alpha_j(u)=\min_{0\le s\le\min\{B_j,u\}}
 \bigl(\alpha_{j-1}(u-s)+\max\{F_j(s),0\}\bigr).
\]
Then $\alpha_h(k)$ is the answer.  Storing a minimizing predecessor recovers an
allocation in either branch.  Each recurrence uses
$O(k\sum_j(B_j+1))=O(hk^2)$ exact transitions.  The next two sections make
this algorithm unconditional for chains by computing every $F_j(s)$ in
closed form.

\section{Closed local tables for chains}
Within component $j$, merge duplicate approval sets by adding their weights.
Retain empty types and their weights in $N_j$, but omit the empty set from
the displayed nonempty chain because it can never strictly support a
deviation.  If a nonempty set remains, write
\[
 A_{j,1}\subsetneq\cdots\subsetneq A_{j,r_j}=D_j\subseteq C_j,
 \qquad D_j=\bigcup_{i\in N_j}A_i;
\tag{8}
\]
otherwise put $D_j=\varnothing$.  Candidates in $C_j\setminus D_j$ are
unapproved padding candidates.  For $0\le s\le B_j$ define the globally
normalized suffix mass
\[
 P_j^{>s}=\sum_{\substack{i\in N_j\\|A_i|>s}}p_i.
\tag{9}
\]

\begin{lemma}[Prefix dominance]\label{lem:prefix}
If $D_j\ne\varnothing$, then for every
$0\le s\le\min\{k,|D_j|\}$ there is a size-$s$ set
$W_{j,s}\subseteq D_j$ satisfying
\[
 u_i(W_{j,s})=\min\{s,|A_i|\}\quad\text{for every local type }i.
\tag{10}
\]
It weakly Pareto-dominates every other size-$s$ local committee and hence
minimizes $\lambda_j$ at that size.
\end{lemma}
\begin{proof}
For $s=0$, take the empty set.  For $s>0$, let $q$ be least with
$|A_{j,q}|\ge s$, set $A_{j,0}=\varnothing$, and choose
\[
 A_{j,q-1}\subseteq W_{j,s}\subseteq A_{j,q},\qquad |W_{j,s}|=s.
\tag{11}
\]
Existence follows from $|A_{j,q-1}|<s\le|A_{j,q}|$.  Lower types are
saturated; higher types approve all of $W_{j,s}$, proving (10).  No size-$s$
committee gives type $i$ more than $\min\{s,|A_i|\}$ utility.  Finally, if
$V_j$ is another size-$s$ committee, then for every deviation $T_j$,
\[
 \{i:u_i(T_j)>u_i(W_{j,s})\}
 \subseteq\{i:u_i(T_j)>u_i(V_j)\}.
\]
Nonnegative weights make every supporter-bearing score against $W_{j,s}$ no
larger than the corresponding score against $V_j$.  Maximizing gives the
claim about $\lambda_j$.
\end{proof}

\begin{theorem}[Closed chain tables]\label{thm:closed}
If $D_j\ne\varnothing$, then
\[
 \boxed{F_j(s)=P_j^{>s}-\frac{s+1}{k}}
 \quad(0\le s<\min\{k,|D_j|\}),
\tag{12}
\]
and
\[
 \boxed{F_j(s)=-\infty}
 \quad(\min\{k,|D_j|\}\le s\le B_j).
\tag{13}
\]
If $D_j=\varnothing$, every table entry is $-\infty$.
\end{theorem}
\begin{proof}
Fix the prefix committee $W_{j,s}$ in the range of (12).  Types with
$|A_i|\le s$ are saturated; every other type has utility $s$.  Hence a
supporter-bearing deviation costs at least $s+1$ and is supported only by
types counted in $P_j^{>s}$.  This gives the upper bound in (12).  Let
$A_{j,q}$ be the smallest chain set of size greater than $s$ and take any
$s+1$ candidates from it.  Every type of approval size greater than $s$
approves the deviation and improves from utility $s$, proving the matching
lower bound.  Lemma~\ref{lem:prefix} makes this the table minimum.

For (13), if $|D_j|\le k$, include all of $D_j$ and pad with null
candidates; every type is saturated.  If $k<|D_j|$, the branch contains only
$s=k$, and a prefix committee gives every unsaturated type utility $k$, the
maximum possible under the deviation cap.  If $D_j$ is empty, no type can
ever improve.  Thus no nonempty local coalition has finite cost.
\end{proof}

The proof includes $s=0$, the boundary $s=k$ or $s=|D_j|$, gaps in chain
cardinalities, zero-weight types, all-zero local mass, empty approvals,
padding candidates, and components larger than $k$.  A positive approval
carried only by zero-weight types yields a finite negative entry in (12),
whereas an empty active union yields $-\infty$.

\begin{example}[Two disjoint approval chains]\label{ex:twochains}
Let $k=2$.  In the first block, weights $1/4,1/4$ approve $\{a\}$ and
$\{a,b\}$; in the second, a weight-$1/2$ type approves $\{c,d\}$.
The suffix formula gives
\[
 (F_1(0),F_1(1),F_1(2))=(0,-3/4,-\infty),
 \quad
 (F_2(0),F_2(1),F_2(2))=(0,-1/2,-\infty).
\]
For instance, $\{a\}$ saturates the first type, while the second must use
$\{a,b\}$ to improve, giving local score $1/4-2/2=-3/4$.
The allocations $(2,0)$, $(1,1)$, and $(0,2)$ have respective values
$0$, $-1/2$, and $0$ in (6).
Thus $W=\{a,c\}$ attains $\Gamma(P)=-1/2$, maximizing the guaranteed
shortfall of deviation support below quota.
Qualitative existence supplies an unblocked committee; this calculation
also selects an optimal committee and quantifies its stability margin.
\end{example}

\section{Exact polynomial-time algorithm}
The input lists candidate indices, approval sets, rational weights in binary,
and the candidate and type partitions.
The algorithm verifies disjointness, local containment, and nesting using
the explicit incidence lists.  Empty types and null candidates
retain their declared block assignments.  Theorem~\ref{thm:closed} removes
local enumeration.  Sort and merge local
ballots, verify adjacent inclusions, bucket weights by approval size, compute
the suffix masses, and order candidates by successive chain layers.  The
first $s$ candidates in this order form $W_{j,s}$; append null candidates in
the saturated padding range.

\begin{theorem}[Exact chain algorithm]\label{thm:algorithm}
For rationally weighted candidate-disjoint chain profiles, the exact value
$\Gamma(P)$ and an attaining committee are computable in time polynomial in
the input encoding length.
\end{theorem}
\begin{proof}
The closed formula constructs every table and a compact prefix argmin, so
the general component DP in Theorem~\ref{thm:dp} already proves the claim.
The chain-specific order-statistic implementation below is sharper and
returns the same value and an attaining allocation directly.  Exact rational
addition and comparison have polynomial bit complexity in the binary input
length; strong polynomiality is not claimed.
\end{proof}

\begin{proposition}[Order statistic and heap reconstruction]
\label{prop:order}
Let
\[
 \mathcal M=\biguplus_{j=1}^h
 \{F_j(0),F_j(1),\ldots,F_j(B_j-1)\},
\tag{14}
\]
where $\sum_jB_j\ge k$ because $\sum_j|C_j|=m\ge k$.  Let
$\theta_{k+1}$ be its $(k+1)$st largest member, with
$\theta_{k+1}=-\infty$ when $|\mathcal M|=k$.  Then
\[
 \boxed{\Gamma(P)=\max\{-1/k,\theta_{k+1}\}.}
\tag{15}
\]
A max-heap that advances the component with the largest current head for
$k$ iterations recovers an attaining seat allocation.
\end{proposition}
\begin{proof}
On the finite branch,
\[
 F_j(s+1)-F_j(s)=-\frac1k-
 \sum_{\substack{i\in N_j\\|A_i|=s+1}}p_i<0,
\]
and Theorem~\ref{thm:closed} gives $F_j(B_j)=-\infty$.  Thus allocating
$s_j$ seats removes the prefix of length $s_j$ from a nonincreasing transition
sequence, and the largest remaining head is $F_j(s_j)$.  Any $k$ prefix
removals leave a head at least $\theta_{k+1}$; repeatedly removing the
largest current head removes the $k$ largest members of $\mathcal M$, so the
greedy allocation attains that threshold.  Prefix closure is automatic.

It remains to justify the bottleneck branch.  Let $d_j$ be the number of
nonnegative members in component $j$'s transition sequence, and let
$\mu_j$ be the mass of its positive-weight nonempty ballots.  If $d_j>0$,
then
\[
 0\le F_j(d_j-1)=P_j^{>d_j-1}-d_j/k,
 \qquad d_j/k\le\mu_j.
\]
Since $\sum_j\mu_j\le1$, at most $k$ transitions are nonnegative.  Hence
$\theta_{k+1}<0$, or it is $-\infty$, and the greedy allocation has only
negative local heads.  Any allocation with a nonnegative head has
$H\ge0$, whereas the greedy one has value
$\max\{-1/k,\theta_{k+1}\}<0$; the factorization therefore proves the
formula.  Ties, exhausted sequences, null padding, and $|\mathcal M|=k$ use
the same convention $F_j(B_j)=-\infty$.
\end{proof}

Initialize the heap with $F_j(0)$ for every $B_j>0$, increment the extracted
component, and insert its next head only while its seat count is below
$B_j$.  Stored prefix committees reconstruct the global committee.
Let $t=|N|+\sum_{i\in N}|A_i|$ count input type records and listed approval
incidences.  Given a validated decomposition with distinct local approval
sets already ordered by inclusion, suffix construction and heap selection
take $O(t+m+(h+k)\log(h+1))$ exact arithmetic/comparison operations.
This bound does not include validation, sorting, or duplicate merging.
Including that preprocessing, the running time is polynomial in the explicit
incidence encoding length and the binary length of the rational weights.
The two-branch DP remains necessary for
arbitrary supplied local tables: nonincreasing abstract tables alone need not
admit an all-nonpositive optimum.

\section{Optimal quantitative bound}
The \emph{positive-support incidence graph} $G_p$ has as vertices the types
$i$ with $p_i>0$ and $A_i\ne\varnothing$ and the candidates approved by at
least one such type; incidence gives its edges.  Thus zero-weight types,
positive-weight empty ballots, and candidates approved by no positive
nonempty type are absent.  Let
\[
 q_{\mathrm{act}}=\text{the number of connected components of }G_p,
\tag{16}
\]
with $q_{\mathrm{act}}=0$ when the graph is empty.  This definition is
intrinsic: reassigning empty ballots, zero-weight types, or null candidates
among declared blocks cannot change it, nor can any metadata coarsening or
refinement.  Whenever the resulting partition is again a valid chain
decomposition, the block characterization below still holds.  In any
nonempty nested chain, a
largest positive-weight approval contains every candidate retained in
$G_p$, so the retained subgraph is connected.  Consequently every valid
chain decomposition has exactly one graph component for each block with
\[
 \mu_j=\sum_{\substack{i\in N_j\\p_i>0,\ A_i\ne\varnothing}}p_i>0.
\]
Isolated candidates still contribute seat capacity, but not support
geometry.

\begin{theorem}[Canonical active-component bound]\label{thm:strong}
If $q_{\mathrm{act}}=0$, then $\Gamma(P)=-1/k$.  If
$q_{\mathrm{act}}\ge1$, then every nonnegative real normalized weight vector
on a candidate-disjoint chain profile satisfies
\[
 \boxed{\Gamma(P)\le
 -\frac{1}{k\min\{q_{\mathrm{act}},k+1\}}.}
\tag{17}
\]
\end{theorem}
\begin{proof}
If $q_{\mathrm{act}}=0$, only zero-weight types can have nonempty approval.
Every deviation therefore has supporter mass zero and score at most $-1/k$;
the no-support singleton baseline attains equality.

Now put $q=q_{\mathrm{act}}$, $r=\min\{q,k+1\}$, and
$\eps=1/(kr)$.  Define the weak demand
$d_j=\min\{s:F_j(s)\le-\eps\}$.  It is finite by (13).  If $\mu_j=0$, then
$d_j=0$: either $D_j=\varnothing$ and $F_j(0)=-\infty$, or only zero-weight
nonempty types remain and $F_j(0)=-1/k\le-\eps$.  If $d_j>0$, minimality and
(12) give, with $Q_j=kr\mu_j$,
\[
 d_j>0\quad\Longrightarrow\quad
 Q_j\ge krP_j^{>d_j-1}>rd_j-1.
\tag{18}
\]
Here $P_j^{>d_j-1}\le\mu_j$ because $d_j-1\ge0$.  Crucially,
$\sum_j\mu_j\le1$, so $\sum_jQ_j\le kr$; positive mass on empty ballots
need not be supportable.

Let $D=\sum_jd_j$ and let $z$ count its positive terms.  Then $z\le q$.
If $D\ge k+1$, summing (18) gives
$kr\ge\sum_{d_j>0}Q_j>rD-z$.  If $q\le k+1$, then $r=q$ and
$rD-z\ge q(k+1)-q=qk=kr$.  If $q\ge k+1$, then $r=k+1$, $z\le D$, and
$rD-z\ge kD\ge k(k+1)=kr$.  Both contradict the strict inequality, so
$D\le k$.

Unite size-$d_j$ local witnesses and pad to $k$ candidates.  Adding
candidates can only shrink strict supporter sets, so every local value stays
at most $-\eps$.  Theorem~\ref{thm:factor}, together with
$-1/k\le-\eps$, proves (17).  The argument includes inert-only blocks,
null candidates, $q_{\mathrm{act}}=1$, and $q_{\mathrm{act}}>k+1$.
\end{proof}

\section{Sharpness and extremal profiles}
For fixed $k,h\ge1$, let $\mathcal P^{\mathrm{act}}_{k,h}$ contain all finite
candidate-disjoint chain profiles with committee size $k$, at most $h$
connected components in $G_p$, and $|C|\ge k$.  Candidate blocks,
capacities, incidence, inert objects, and weights may vary within this
extremal class.

\begin{theorem}[Exact extremal value]\label{thm:extremal}
The maximum is over all profiles in $\mathcal P^{\mathrm{act}}_{k,h}$,
allowing candidate count, blocks, incidence, and normalized weights to vary.
 Put $r=\min\{h,k+1\}$.  Then
\[
 \boxed{\max_{P\in\mathcal P^{\mathrm{act}}_{k,h}}\Gamma(P)=-\frac1{kr}.}
\tag{19}
\]
\end{theorem}
\begin{proof}
If $q_{\mathrm{act}}=0$, the value is $-1/k\le-1/(kr)$.  Otherwise
$q_{\mathrm{act}}\le h$ and Theorem~\ref{thm:strong} gives
\[
 \Gamma(P)\le-\frac1{k\min\{q_{\mathrm{act}},k+1\}}
 \le-\frac1{kr},
\]
with the displayed sign direction.  For $r=1$, a weight-one type approving
one candidate, together with $k-1$ null candidates, is a one-component
witness of value $-1/k$; the deviation cap is why the composition argument
below is not used.  Suppose $r\ge2$.  Choose positive integers
$d_1+\cdots+d_r=k+1$, give component $j$ exactly $d_j$ approved candidates,
and give its full-set type weight
\[
 p_j=\frac{d_j}{k}-\frac1{kr}=\frac{rd_j-1}{kr}.
\tag{20}
\]
Every approval is nonempty, the weights are positive and sum to one, and
$G_p$ is the disjoint union of exactly $r$ connected stars.  Thus the
profile belongs to the revised class without counting an empty or zero-weight
type.  Any size-$k$ committee
underallocates some component below $d_j$, whose $d_j$-candidate deviation
scores $-1/(kr)$.  Conversely, allocate $d_j$ seats to every component but
one, and $d_j-1$ to the remaining component.  Its best local value is
$-1/(kr)$; every other local value is at most $-1/(kr)-1/k$ or
$-\infty$.  Factorization gives equality.  At $r=k+1$, this is the uniform
singleton endpoint with weight $1/(k+1)$.
\end{proof}

For $m\ge k\ge1$ and $h\ge1$, let
$\mathcal P^{\mathrm{act}}_{m,k,h}$ be the subfamily with exactly $m$
candidates.

\begin{theorem}[Fixed-candidate extremum]\label{thm:fixedm}
Fix $m,k,h$; candidate blocks, incidence, and normalized weights remain
variable within $\mathcal P^{\mathrm{act}}_{m,k,h}$.  Then
\[
 \boxed{\max_{P\in\mathcal P^{\mathrm{act}}_{m,k,h}}\Gamma(P)=
 \begin{cases}
  -1/k,&m=k,\\[2mm]
  -1/[k\min\{h,k+1\}],&m\ge k+1.
 \end{cases}}
\]
\end{theorem}
\begin{proof}
If $m=k$, the unique committee is $C$.  For every ballot and every
$T\subseteq C$, $|A_i\cap T|\le|A_i|=|A_i\cap C|$, so no voter strictly
improves; the singleton no-support baseline gives $\Gamma=-1/k$.

Let $m\ge k+1$ and $r=\min\{h,k+1\}$.  Theorem~\ref{thm:strong} gives the
upper bound.  If $r\ge2$, use the preceding construction on its exactly
$k+1$ active candidates and assign the remaining $m-(k+1)$ null candidates
to existing blocks.  This changes neither $G_p$ nor the attaining committee,
and any size-$k$ committee still omits an active candidate.  If $r=1$, use
one weight-one type approving exactly $k$ candidates and make the remaining
$m-k$ candidates null.  Selecting all approved candidates attains $-1/k$.
\end{proof}

For a subfamily with fixed candidate blocks or incidence, the upper bound
still holds, while equality follows if that subfamily contains the displayed
witness for some positive composition $d_1+\cdots+d_r=k+1$.  This is a
sufficient witness condition, not a classification of necessary capacities.
For $r=1$, no extra active-candidate capacity is required beyond the model.

\begin{corollary}[Exactly $h$ active components]\label{cor:exacth}
For fixed $k,h\ge1$, let $\mathcal E_{k,h}$ consist of all finite
candidate-disjoint chain profiles with $|C|\ge k$ and $q_{\mathrm{act}}=h$.
Candidate count, blocks, incidence, and normalized weights may vary.  Then
\[
 \sup_{P\in\mathcal E_{k,h}}\Gamma(P)
 =-\frac1{k\min\{h,k+1\}}.
\]
For $h\le k+1$, the supremum is attained.  For $h>k+1$, the family below
approaches it; nonattainment in the full class is not asserted.
\end{corollary}
\begin{proof}
Theorem~\ref{thm:strong} gives the upper bound.  When $h\le k+1$, the
construction in Theorem~\ref{thm:extremal} has exactly $h$ active components
and attains equality.
For $h>k+1$, put $\ell=h-k-1$.  Take $k+1$
principal singleton components of mass
\[
 a_\delta=\frac{1-\ell\delta}{k+1}
\]
and $\ell$ additional singleton components of mass $\delta$, where
$0<\delta\le1/h$.  Every singleton mass is below $1/k$, and
$a_\delta\ge\delta$.  A best deviation against a committee is therefore a
largest-weight unselected singleton.  At least one principal is unselected,
and selecting $k$ principals attains
\[
 \Gamma(P_\delta)=a_\delta-\frac1k
 =-\frac1{k(k+1)}-\frac{\ell\delta}{k+1}.
\tag{21}
\]
Every ballot is nonempty and positive, so $G_p$ has exactly $h$ one-edge
components.  Hence the exact supremum on the class with exactly $h$
support-active components is $-1/[k(k+1)]$.
\end{proof}

\begin{lemma}[Continuity]\label{lem:continuity}
For fixed finite incidence and fixed $k$, the map $p\mapsto\Gamma(P)$ is
continuous and piecewise linear.
\end{lemma}
\begin{proof}
For fixed $W,T$, the score is affine in $p$.  The value for $W$ is a finite
maximum of affine functions, and $\Gamma$ is a finite minimum over
committees.  Finite minima and maxima preserve continuity and piecewise
linearity.
\end{proof}

\begin{theorem}[Total-variation stability]\label{thm:tv}
For any two probability distributions $p,p'$ on the full ballot space
$2^C$, with the same candidate set and committee size,
\[
 \boxed{|\Gamma(p)-\Gamma(p')|\le\operatorname{TV}(p,p')},
 \qquad
 \operatorname{TV}(p,p')=
 \sup_{E\subseteq2^C}|p(E)-p'(E)|
 =\frac12\lVert p-p'\rVert_1.
\]
No component, chain, or common-support assumption is required.
\end{theorem}
\begin{proof}
For fixed $W,T$, the strict supporters form the event
$E_{W,T}=\{A\subseteq C:|A\cap T|>|A\cap W|\}$.  Hence the two deviation
scores differ by at most $\operatorname{TV}(p,p')$.  For finite families,
\[
 |\max_x a_x-\max_x b_x|\le\max_x|a_x-b_x|,
 \qquad
 |\min_x a_x-\min_x b_x|\le\max_x|a_x-b_x|.
\]
Taking first the maximum over deviations and then the minimum over committees
proves the claim.
\end{proof}

If $\Gamma(p)\le-\delta$ and
$\operatorname{TV}(p,p')<\delta$, then
$\Gamma(p')\le\Gamma(p)+\operatorname{TV}(p,p')<0$, so $p'$ has nonempty
core.  In particular, Theorem~\ref{thm:strong} gives the explicit radius
$\delta=1/[k\min\{q_{\mathrm{act}},k+1\}]$ around any chain profile with
$q_{\mathrm{act}}\ge1$, even when the perturbed profile is not a chain.  No
optimality of this radius is asserted.

Thus (19) is a maximum on the capacity-unrestricted closed class
$\mathcal P^{\mathrm{act}}_{k,h}$.  Corollary~\ref{cor:exacth} gives the
exactly-$h$ supremum, with the approximating family (21) available in that
varying-profile class.  Theorem~\ref{thm:fixedm} records the exact fixed-candidate
boundary, including the forced baseline at $m=k$.

\section{Candidate intervals and related work}
Within each nonempty active chain, order candidates by successive difference
layers and append $C_j\setminus D_j$; order an all-null block arbitrarily.
Every local approval set is a prefix of its block.  Concatenating the blocks
makes every approval set an interval in one global candidate order, with the
empty set an interval vacuously.  The present chain domain is therefore a
subclass of candidate-interval preferences.

Pierczy\'nski and Skowron define candidate-interval preferences and prove
polynomial-time core nonemptiness for finite profiles
~\cite{pierczynski2021}.  Qualitative nonemptiness of our domain was thus
known, and follows for rational type weights by replication; it is not a new
claim.  Their theorem does not state the signed local table (12), the exact
factorization (4), the canonical active-component constant (17), or the
total-variation radius obtained from that constant.

Berker et al. formulate the sizewise least core~\cite{berker2026}.  Their
Theorem~2 supplies the equal-singleton construction for the outer extremum,
while their Theorems~3 and~5 supply the $m=k+1$ upper side.  Together these
results already contain the $h=k+1$, $m=k+1$ endpoint of
Theorems~\ref{thm:extremal} and~\ref{thm:fixedm}, which we do not relabel as
new; they do not state our support-active refinement for arbitrary $m$.
Approval-based apportionment permits repeated party seats and proves core
existence and polynomial computation in that model~\cite{brill2024}; its
clone interpretation does not yield our local table.  Ferrers, or chain,
graphs encode bipartite graphs with neighborhoods ordered by inclusion
~\cite[Definition~1]{gyorffy2025}.  Recent work studies Pareto-optimal
committees under candidate intervals~\cite{schunke2026}.  Polynomial-time
Thiele optimization on voter-interval profiles concerns a different interval
orientation and objective~\cite{manurangsi2026}.  These results do not give
the universal utility vector (10) or its strict-support least-core
consequence.  Parameter results for
$k\le8$, $m\le15$, and at most seven voter types are incomparable with the
chain hypothesis~\cite{peters2025,becker2026}.  Laminar profiles and
acyclic voter--candidate incidence impose different structure
~\cite{petersskowron2022,haret2024}.  Standard knapsack and min-plus
recurrences explain the Bellman mechanics~\cite{cygan2017,bateni2018}.
Exact least-core algorithms are known for particular transferable-utility
games~\cite{hatano2017,fang2016}, while additive game compositions can have
Cartesian-product cores~\cite{solymosi2024}.  Their payoff variables and
inessential cross-component coalitions differ from globally seat-capped
approval deviations, so they do not imply (3), (12), or (17).  The
order-statistic extraction is a standard merge of decreasing sequences once
the chain-specific nonpositive reduction has been proved; the
total-variation step is the standard stability of finite maxima and minima.

\section{Exact regression checks}
Example~\ref{ex:twochains} connects the local tables to an explicit optimal
seat allocation.  Two independently implemented exact checkers provide
complementary finite regression tests of the general results.
The strengthened checker compares local tables with (12)--(13), global
values with factorization, and allocations with both DP branches.
The focused checker tests the intrinsic support graph, extremal witnesses,
heap reconstruction, and total-variation stability.
Both use exact rational arithmetic and explicitly test empty ballots,
zero weights, null candidates, and saturation at the deviation cap.
These checks supplement, but do not replace, the analytic proofs.

\begin{table}[ht]
\centering\small
\begin{tabular}{@{}lrrrrr@{}}
\toprule
Profile & $m$ & $k$ & $n$ & $q_{\mathrm{act}}$ & $\Gamma$ \\
\midrule
A & 16 & 10 & 50 & 5 & $-1/50$ \\
B & 16 & 10 & 80 & 8 & $-1/80$ \\
C & 16 & 11 & 16 & 4 & $-1/44$ \\
\bottomrule
\end{tabular}
\caption{Exact values for three regression benchmarks.}
\end{table}

All three benchmarks attain the active-component bound (17).
Appendix~\ref{app:regression} reports the full score histograms and the
coverage of each checker.  Complete ballot lists are embedded in both
checker sources.

\section{Limitations and the first non-chain fork}
The closed formula depends on the prefix-dominance order.  It does not apply
to arbitrary candidate-interval profiles, overlapping components, signed
weights, weak improvement, Droop quota, non-unit candidate costs, or
deviations larger than $k$.  The extremal maximum is over the stated
varying-profile universe; in a fixed-incidence subfamily equality needs an
applicable sufficient witness.  Theorem~\ref{thm:fixedm} resolves fixed
candidate count, but not fixed incidence; in particular, $|C|=k$ forces the
baseline value $-1/k$.

A connected non-chain obstruction occurs with two positive-weight approvals
$\{a,b\}$ and $\{a,c\}$.  Their incidence graph is connected, so they cannot
be separated into candidate-disjoint chain blocks.
At two local seats, no committee simultaneously gives both types utility two,
and the prefix-dominance mechanism fails.

\begin{proposition}[Canonical rooted-fork baseline]\label{prop:fork}
Let $\varnothing\ne R$, let $C=R\sqcup\{b,c\}$, and suppose every ballot is
one of $\varnothing$, $R$, $R\cup\{b\}$, and $R\cup\{c\}$.  For every
nonnegative normalized weight vector and every $1\le k\le|C|$,
\[
 \Gamma(P)=-1/k.
\]
\end{proposition}
\begin{proof}
The no-support singleton gives the universal lower bound $-1/k$.  If
$k\le|R|$, choose any $k$ candidates from $R$.  Every nonempty type then has
utility $k$, which no admissible deviation can improve.  If $k=|C|$, choose
$C$ and every type is saturated.

It remains that $k=|R|+1\ge2$.  Let $x_b$ and $x_c$ be the total weights on
the two branch ballots.  Relabel the branches so that $x_c\le x_b$, hence
$x_c\le1/2$, and choose $W=R\cup\{b\}$.  Only the $c$-branch type can improve,
and it can do so only through $T=R\cup\{c\}$, whose score is $x_c-1$.
Therefore $g_p(W)=\max\{-1/k,x_c-1\}=-1/k$.
\end{proof}

Proposition~\ref{prop:fork} crosses the first non-chain fork only for two one-candidate
extensions of a common nonempty root.  It does not restore the chain table or
analyze arbitrary branch depths, overlaps, or additional incomparable sets.

The remaining fork project starts with arbitrary
$A_0\subsetneq A_1,A_2$ and incomparable branches.  The results neither
construct an empty-core election nor settle unrestricted
core nonemptiness.  They do not quantitatively solve all candidate-interval
profiles or the eight-type frontier.

\section{Conclusion}
We give a complete exact analysis of the sizewise-least-core value on the
stated candidate-disjoint nested-chain domain.  Local
scores factor, prefix-packed committees close every local table, and a
transition order statistic with heap reconstruction returns the exact
rational optimum and committee in polynomial time.  The same table proves
the optimal $1/[k\min\{q_{\mathrm{act}},k+1\}]$ per-profile guarantee.
Capacity-matched profiles attain the at-most-$h$ extremal value, the fixed-$m$
theorem identifies the sole $m=k$ branch, and an explicit positive family
determines the exactly-$h$ supremum beyond $k+1$ components.  Total-variation
stability promotes each negative margin to a neighborhood of core-nonempty
profiles.  Known candidate-interval theory supplies the qualitative context;
the contribution here is the exact quantitative decomposition.

\appendix
\section{Proof dependencies}
Theorem~\ref{thm:coalition} supplies the no-support baseline.
Theorems~\ref{thm:factor} and~\ref{thm:dp} use finite attainment,
disjointness, over-budget domination, and monotonicity.
Lemma~\ref{lem:prefix} supplies Theorem~\ref{thm:closed}; that closed table
drives Theorem~\ref{thm:algorithm}, Proposition~\ref{prop:order}, and the
demand contradiction in Theorem~\ref{thm:strong}.
Theorems~\ref{thm:extremal} and~\ref{thm:fixedm} add explicit capacity
witnesses; Corollary~\ref{cor:exacth} records the exactly-$h$ consequence.
Lemma~\ref{lem:continuity} controls boundary limits, and
Theorem~\ref{thm:tv} supplies structure-free perturbation stability.
Proposition~\ref{prop:fork} is a separate boundary argument using the
no-support baseline and direct utility saturation.
The public verification files record the finite checks and their scopes.

\section{Exact regression data and reproduction}\label{app:regression}
For each benchmark, the integer margin is $M(W)=nkg_p(W)$.
Run-length encoding (RLE) writes each margin followed by
its multiplicity as \texttt{value:count}.  For example,
\texttt{-10:5940} denotes 5,940 committees with integer margin $-10$.
\begin{table}[ht]
\centering\small
\begin{tabular}{@{}lrrrrl@{}}
\toprule
Profile & $m$ & $k$ & $n$ & $\Gamma$ & complete margin RLE \\
\midrule
A & 16 & 10 & 50 & $-1/50$ & $-10{:}5940,40{:}1980,80{:}22,90{:}66$ \\
B & 16 & 10 & 80 & $-1/80$ & $-10{:}5202,70{:}2610,140{:}195,210{:}1$ \\
C & 16 & 11 & 16 & $-1/44$ & $-4{:}3264,12{:}1056,28{:}48$ \\
\bottomrule
\end{tabular}
\caption{Complete exact regression histograms.}
\end{table}

The complete ballot lists and multiplicities are embedded in both checker
sources.  The strengthened checker uses \texttt{BENCHMARKS}; the focused
checker uses \texttt{benchmark\_profiles()}.

Benchmark A has five support-active chain components and attains
$-1/(10\cdot5)$.  Benchmark B consists of eight full-set components whose
closed tables encode demands $(2,2,2,1,1,1,1,1)$ and attains
$-1/(10\cdot8)$.  Benchmark C has four two-level chains and attains
$-1/(11\cdot4)$.  Thus all three are exact examples of (12), (17), the
allocation DP, and Proposition~\ref{prop:order}, rather than evidence for the
analytic proofs.

The public strengthened checker is a stand-alone implementation that imports
no manuscript proof module.  It uses \texttt{frozenset} objects, combination enumeration, and
\texttt{fractions.Fraction}.  It directly
compares brute local tables with (12)--(13), compares brute global values with
factorization and both DP branches, and checks the strengthened bound,
sharpness, continuity samples, exact positive-component approximation, and
Benchmarks A--C.  Its current bound check excludes declared blocks containing
only empty or zero-weight ballots from $q_{\mathrm{act}}$.  Its declared
domains cover $k=1,\ldots,4$,
$h=1,\ldots,5$, chain lengths through four, gaps, duplicates, empty and
zero-weight types, zero-mass components, null padding, and components larger
than $k$.  These finite checks are falsification and regression evidence
only.
The sealed run passed 41,556 exact assertions.  From the
\texttt{verification} directory, the canonical commands are
\texttt{python strengthened\_tests.py} and
\texttt{python strengthened\_checker.py}.  Its public sealed JSON is
deterministic and byte-reproducible.

The public focused checker uses integer
bit masks and a separately implemented brute-force objective.  It tested the
intrinsic support graph, sharp
active-component and fixed-$m$ witnesses, the order statistic and heap
reconstruction, and total-variation stability.  It passed 351,647 exact
assertions, including 1,913 small chain-profile/committee-size cases,
224,149 exhaustive distribution pairs, 87,669 strict robust-core
implications, and exact heap witnesses for Benchmarks A--C.  Its eleven unit
tests in \path{minimal_revision_tests.py} also pass.  The focused checker
directly tests the support graph, Theorems~\ref{thm:strong}--\ref{thm:tv}, and
Proposition~\ref{prop:order}; the strengthened checker primarily tests the
local tables, factorization, DP, and Theorems~\ref{thm:algorithm}--
\ref{thm:extremal}.  Their coverage overlaps but is not identical, and all
finite checks remain regression evidence.

\section*{Data and code availability}
All code, exact verification artifacts, and reproducibility instructions
supporting this paper are available at
\url{https://github.com/Baymax-ray/sizewise-least-core}.

\section*{Use of AI-assisted tools}
\begingroup\emergencystretch=1em
OpenAI GPT-5.6 Sol was used during
exploratory proof search, computational prototyping, literature-query
formulation, and manuscript revision.  All statements, proofs, citations,
and computational artifacts included in the final manuscript were reviewed
by the author.  Exact computational claims are accompanied by independently
implemented verification code.  The author takes full responsibility for
the content of the paper.
\par\endgroup


\begin{thebibliography}{16}
\footnotesize
\setlength{\itemsep}{0pt}
\bibitem{pierczynski2021}
G.~Pierczy\'nski and P.~Skowron,
\emph{Core-Stable Committees under Restricted Domains},
Web and Internet Economics (WINE 2022), LNCS 13778, 311--329, 2022,
doi:10.1007/978-3-031-22832-2\_18.
\bibitem{berker2026}
R.~E.~Berker, E.~Tewolde, V.~Conitzer, M.~Guo, M.~Heule, and L.~Xia,
\emph{On the Edge of Core (Non-)Emptiness: An Automated Reasoning Approach
to Approval-Based Multi-Winner Voting}, AAAI 40(20), 16673--16681, 2026;
doi:10.1609/aaai.v40i20.38709.
\bibitem{peters2025}
D.~Peters, \emph{The Core of Approval-Based Committee Elections with Few
Seats}, Proceedings of IJCAI 2025, 4014--4022, 2025,
doi:10.24963/ijcai.2025/447; full version: arXiv:2501.18304v2.
\bibitem{becker2026}
P.~Becker, M.~Greger, and D.~Peters,
\emph{Core Existence in Approval-Based Committee Elections with up to Seven
Voter Types}, arXiv:2605.06194v2, 2026.
\bibitem{brill2024}
M.~Brill, P.~G\"olz, D.~Peters, U.~Schmidt-Kraepelin, and K.~Wilker,
\emph{Approval-based apportionment}, Mathematical Programming 203 (2024),
77--105.
\bibitem{cygan2017}
M.~Cygan, M.~Mucha, K.~W\k{e}grzycki, and M.~W\l{}odarczyk,
\emph{On Problems Equivalent to (min,+)-Convolution}, arXiv:1702.07669, 2017.
\bibitem{bateni2018}
M.~H.~Bateni, M.~T.~Hajiaghayi, S.~Seddighin, and C.~Stein,
\emph{Fast Algorithms for Knapsack via Convolution and Prediction},
arXiv:1811.12554, 2018.
\bibitem{haret2024}
A.~Haret, S.~Klumper, J.~Maly, and G.~Sch\"afer,
\emph{Committees and Equilibria: Multiwinner Approval Voting Through the
Lens of Budgeting Games}, EC 2024, 51--70.
\bibitem{petersskowron2022}
D.~Peters and P.~Skowron,
\emph{Proportionality and the Limits of Welfarism},
arXiv:1911.11747v3, 2022.
\bibitem{schunke2026}
J.~Sch\"unke,
\emph{Pareto Optimality in Approval-Based Multiwinner Voting},
arXiv:2605.30490v1, 2026.
\bibitem{gyorffy2025}
L.~Gy\H{o}rffy, A.~London, G.~V.~Nagy, and A.~Pluh\'ar,
\emph{Partitions of Graphs into Special Bipartite Graphs},
Discrete Mathematics \& Theoretical Computer Science 27(3), 2025;
doi:10.46298/dmtcs.15361.
\bibitem{manurangsi2026}
P.~Manurangsi and K.~Sornat,
\emph{Polynomial-Time Algorithm for Thiele Voting Rules with Voter Interval
Preferences}, arXiv:2604.05953v1, 2026.
\bibitem{hatano2017}
D.~Hatano and Y.~Yoshida,
\emph{Computing Least Cores of Supermodular Cooperative Games},
AAAI 31(1), 2017; doi:10.1609/aaai.v31i1.10583.
\bibitem{fang2016}
Q.~Fang, B.~Li, X.~Sun, J.~Zhang, and J.~Zhang,
\emph{Computing the Least-Core and Nucleolus for Threshold Cardinality
Matching Games}, Theoretical Computer Science 609, 500--510, 2016;
doi:10.1016/j.tcs.2015.11.015.
\bibitem{solymosi2024}
T.~Solymosi,
\emph{Assignment Games with Population Monotonic Allocation Schemes},
Social Choice and Welfare 62, 67--88, 2024;
doi:10.1007/s00355-023-01477-z.
\end{thebibliography}
\end{document}